\documentclass[11pt]{article}
\usepackage{graphicx} 
\usepackage{amsmath,amssymb,enumerate,latexsym}
\usepackage{hyperref}
\usepackage{thm-restate}
\usepackage{cleveref}
\usepackage{authblk}
\usepackage[letterpaper, total={6.5in, 9in}]{geometry}
\usepackage{braket}

\newcommand{\assign}{:=}
\newcommand{\nin}{\not\in}
\newcommand{\op}[1]{#1}
\newcommand{\tmop}[1]{\ensuremath{\operatorname{#1}}}

\newenvironment{enumeratenumeric}{\begin{enumerate}[1.] }{\end{enumerate}}
\newenvironment{enumerateroman}{\begin{enumerate}[i.] }{\end{enumerate}}
\newenvironment{itemizedot}{\begin{itemize} }{\end{itemize}}
\newenvironment{proof}{\noindent\textbf{Proof\ }}{\hspace*{\fill}$\Box$\medskip}

\newtheorem{theorem}{Theorem}[section]
\newtheorem{lemma}[theorem]{Lemma}
\newtheorem{definition}[theorem]{Definition}

\newtheorem{corollary}[theorem]{Corollary}
\newtheorem{conjecture}[theorem]{Conjecture}

\title{Quantum Fine-Grained Lower Bounds for SetDisjointness via Sub-Linear Reductions from 3SUM}
\author{Jeremy Ahrens Huang}
\author{Young Kun Ko}
\author{Chunhao Wang}
\affil[]{Department of Computer Science and Engineering, Pennsylvania State University}
\affil[]{Email: \{jeremyah,ykko,cwang\}@psu.edu}
\date{September 2026}

\begin{document}

\maketitle

\begin{abstract}
    In classical fine-grained complexity, the 3SUM Conjecture is used to prove a variety of conditional lower bounds on data structure and graph problems via an initial reduction to the SetDisjointness problem. However, there is an $\tilde{O}(n)$-time quantum algorithm for 3SUM and a direct application of Grover's algorithm to SetDisjointness queries beats the state-of-the-art classical conditional bound by Kopelowitz, Pettie, and Porat (SODA 2016); this shows that these classical bounds do not apply in the quantum setting. Thus establishing analogous conditional lower bounds in the quantum setting requires applying the quantum 3SUM Conjecture to a \emph{quantum} fine-grained reduction from 3SUM to SetDisjointness.

    We give the first sub-linear time quantum reductions from 3SUM to online SetDisjointness. Via our reduction, the quantum 3SUM conjecture implies a $p + 2q \geqslant 1$ tradeoff bound for quantum SetDisjointness algorithms with $O(N^p)$ preprocessing time and $O(N^q)$ query time. We also give an analogous reduction from 3XOR. These results are derived from a general framework for fine-grained reductions to SetDisjointness which applies to any Abelian 3-Orthogonal Array (3OA) problem with suitable almost-linear hash functions.
\end{abstract}

\section{Introduction}

In classical fine-grained complexity, conditional lower bounds for many data structure and graph problems are given via reduction from the 3SUM problem through the SetDisjointness problem using the following state-of-the-art bound by \cite{KPP16}: if 3SUM requires $\Omega(n^{2-o(1)})$ classical time then for any $\gamma \in (0,1)$ answering $\tilde{O}(n^{1+\gamma})$ SetDisjointness queries between sets of size $O(n^{1-\gamma})$ must take $\Omega(n^{2-o(1)})$ classical time. Unfortunately, both sides of this bound are broken in the quantum setting: 3SUM is known to be in $\tilde{O}(n)$ quantum time, and even a naive application of Grover search to answer each SetDisjointness query in $O(n^{(1-\gamma)/2})$ quantum time solves these SetDisjointness instances in sub-quadratic time. Since 3SUM is in $\tilde{O}(n)$ quantum time, analogous conditional data structure and graph lower bounds for quantum fine-grained complexity would require a sub-linear-time reduction from 3SUM to SetDisjointness. This is a tall order for existing reduction strategies which rely on steps, such as reading the responses to all $\tilde{O}(n^{1+\gamma})$ queries, that are inherently super linear.

In this paper we show the first quantum fine-grained reduction from 3SUM to online SetDisjointness that has sub-linear running time via a general framework for sub-linear time reductions to SetDisjointness from all Abelian 3-Orthogonal Array problems (in both quantum and classical settings). We also apply our framework to obtained a sub-linear time reduction from 3XOR to online SetDisjointness, and we think it will be helpful for proving reductions from other problems as well. In order to create our framework we also made two novel technical contributions which we think will be useful for many other quantum algorithms and reductions. First, we completed a general time complexity analysis of the $L$-subset finding quantum walk \cite{CE05, Amb04}; previous works were focused only on its query complexity. Second, we give a tighter error analysis of \cite{Amb04, BLPS22}’s skip-list data structure and give general lemmas for its properties; previously it was only analyzed as a part of the $L$-subset finding quantum walk. 

\subsection{Previous work}
The classical 3SUM conjecture is one of the three most popular conjectures in fine-grained complexity \cite{Vas15}, and has withstood over thirty years of study since it was first proposed as an underlying barrier for geometric problems \cite{GO95}. The conjecture is that the 3SUM problem requires $\Omega(n^{2-o(1)})$ classical time, or equivalently, that there does not exist a constant $\delta>0$ such that 3SUM is in $O(n^{2-\delta})$ classical time; the quantum 3SUM conjecture \cite{AL20} similarly states that the 3SUM problem requires $\Omega(n^{1-o(1)})$ quantum time. The 3SUM problem is the problem of deciding if, out of a list $S$ of $n$ integers from $[-n^3,\ldots,n^3]$, there are three $(a,b,c)$ such that $a+b=c$. 3SUM can be solved by a naive search over all unique unordered pairs $(i,j) \in [n]$ for one such that $S_i+S_j \in S$, which can be done in $\tilde{O}(n)$ quantum time and $\tilde{O}(n^2)$ classical time, so the 3SUM conjectures implicitly ask whether it is possible to do meaningfully better than naive search.

\cite{Pat10} showed that disproving the classical 3SUM conjecture is also a barrier for graph and data structure problems by reducing the 3SUM problem to SetIntersection (a variant of SetDisjointness) in sub-quadratic time via a 3SUM variant called Conv3SUM. \cite{KPP16} subsequently proved tighter conditional lower bounds (CLBs) on SetDisjointness and SetIntersection by giving a direct reduction from 3SUM (skipping Conv3SUM). SetDisjointness is the problem of deciding if the intersection of two sets is empty for each of $q$ queries between sets drawn from a given family of sets $\mathcal{F}$ with elements drawn from a universe $U$; the input size $N$ is given by the total of the number of elements in each set $\sum_i |\mathcal{F}_i|$. Note that this problem cannot be a regular decision or search problem because the output is potentially larger than the input since $q$ could be as large as $|\mathcal{F}|^2$. SetIntersection is a variant of SetDisjointness where the output for each query is the full contents of the intersection between sets.

Since 3SUM can be solved in linear quantum time, analogous reductions in the quantum setting must be \emph{sub-linear} time. There are two main impediments to sub-linear versions of these reductions: first, almost all of them require placing all $n$ integers of $S$ into some kind of data structure, second, the reductions like those of \cite{Pat10} and \cite{KPP16} require their destination problems to produce \emph{super-linear} outputs. \cite{BLPS22} tackled the first impediment by using the quantum walk and data structure of \cite{Amb04} to implement the data structures needed by \cite{GO95} in sub-linear. They were also able to quantize the reduction from 3SUM to Conv3SUM by modifying the data structure to support index operations. This work is the first to address the second impediment. 

While the time complexity of 3SUM is an ongoing area of study, the quantum query complexity of 3SUM is very well understood. The general $L$-subset finding quantum walk algorithm of \cite{CE05, Amb04} (the same one used by \cite{BLPS22}) finds any $L$ items satisfying any predicate in a set of $n$ items in $O(n^\frac{L}{L+1})$ queries and works as a kind of self reduction to the $L$-subset finding problem with a data structure on a sub-linear size subset of the original problem. Thus it can solve to 3SUM in $O(n^{3/4})$ queries. Surprisingly, \cite{BS13, Spa13a} found that this general upper bound is tight for 3SUM; they find that the only property of 3SUM that matters for its query complexity is that for every pair of numbers there is exactly one number such that the three numbers together make a valid triple. They generalize this as a $\Theta(n^\frac{k}{k+1})$ query bound for all ``full-strength'' $k$-orthogonal array problems (the class of problems where all $(k-1)$ tuples of elements have exactly one additional element that makes them valid). 

\subsection{Results}

We start with our main result: a preprocessing and query time tradeoff for online SetDisjointness 
conditioned on the hardness of 3SUM or 3XOR. 

\begin{restatable}[]{theorem}{restateonlinetheorem}
  Unless the quantum 3SUM conjecture is false and
  there is an $O (n^{1 - \delta})$-time quantum algorithm for 3XOR with $\delta > 0$,
  for any constants $p \in \left[ \frac{1}{2}, 1 \right)$ and $q \in \left[
  0, \frac{1}{4} \right]$ such that there is a quantum algorithm for 
  Online SetDisjointness with $O (N^p)$
  preprocessing time and $O (N^q)$ time per query which correctly
  answers queries with bounded error the following inequality must hold
  \[p + 2 q \geqslant 1.\] 
  
\end{restatable}

We restate \cite{KPP16}'s classical tradeoff bound below for comparison: 
\begin{theorem}[\cite{KPP16} Corollary 1.8 rephrased]
    Assume the classical 3SUM conjecture.
    Fix constants $p \in [1,2)$ and $q \in [0, 1/2]$ and suppose there is 
    an algorithm for SetDisjointness with $O (N^p)$ classical
    preprocessing time and $O (N^q)$ classical time per query. Then 
    \[p+2q\geqslant 2.\] 
\end{theorem}

Our tradeoff bound follows from our framework for obtaining reductions from 3OA problems 
to online SetDisjointness.

\begin{restatable}[CLB framework for online SetDisjointness]{theorem}{restateonlineframeworktheorem}
  \label{thm:clb-framework-online}
  Suppose:
  
  - for any $\gamma \in (0, 1)$ there is an efficient family of
  hash functions $h_1 : G \rightarrow R$ with $| R | = n^{\gamma}$ that is
  $n^{\beta}$-almost balanced with $\beta \in (0, 1)$ for size factor $c''$
  and $\tilde{O} (1)$-almost linear for the abelian group $(G, +_G)$ in the
  operation $+_R$, and
  
  - for any $\gamma \in (0, 1)$ there is an efficient family of
  hash functions $h_2 : G \rightarrow Q$ that is pairwise independent and
  $\tilde{O} (1)$-almost linear for $(G, +_G)$ in the operation $+_Q$, where
  $| Q | = O \left( {c''}^2 n^{2 - 2 \gamma} \right)$ and $Q$ has an efficient
  P{\u a}trașcu Split of size $O \left( \sqrt{Q} \right)$,
  
  and fix some constants $p \in \left[ \frac{1}{2}, 1 \right)$ and $q \in \left[
  0, \frac{1}{4} \right]$.
  Then any algorithm for Online SetDisjointness with $O (N^p)$ quantum
  preprocessing time and $O (N^q)$ quantum time per query which correctly
  answers queries with bounded error must have $p + 2 q \geqslant 1$ unless
  there is an $O (n^{1 - \delta})$-time bounded-error algorithm for the 3OA 
  problem on the abelian group $(G, +_G)$ with no false positives for 
  some constant $\delta > 0$.
\end{restatable}

Note that since online SetDisjointness allows to know which queried set pairs 
are non-disjoint we can search the corresponding buckets for a 3OA witness so 
we can rule out false positives and thus give an algorithm for 3OA with 
one-sided error.

We also made two novel technical improvements that
we think will be useful for other quantum algorithms and reductions. First, we 
give the first general time complexity analysis of the $L$-subset finding quantum walk
that we know of in the literature.
\begin{lemma}[Informal version of \Cref{lem:time-of-johnson-walk}]
    For any $\delta \in \mathbb{R}$ if there is an 
$\tilde{O}(n^{1-\delta})$ time quantum algorithm for finding a specific 
type of $L$-subset with a suitable data structure, then
the $L$-subset finding quantum walk of \cite{CE05} for that type of subset
has an overall time complexity 
of $\tilde{O}(\max(N^{\frac{L}{L+2\delta}}, N^\frac{L}{L+1}))$.
\end{lemma}
Also see \Cref{cor:sublinear-walk}. Second, we give a much tighter error analysis for the skip-list data structure of \cite{Amb04, BLPS22}.
\begin{lemma}[Extracted from \Cref{lem:skip-data-structure}]
    The probability that a single operation on the data structure for \cite{Amb04, BLPS22} fails can be made to be $O(1/n^d)$ for any constant $d>0$.
\end{lemma}
The previous analysis only gave an $O(1/n^4)$ bound on the error probability. We also explicitly state the possible operations, time complexities, space complexity, and cumulative error probability of the data structure when used in any algorithm for any general data type in \Cref{lem:skip-data-structure} and \Cref{lem:skip-data-structure-cumulative-error}. Previous works only analyzed these data structures when used to store integers as a part of a particular $L$-subset finding quantum walk. 

\subsection{Reduction strategy and techniques}

We provide reductions from 3OA on groups with good hash functions that preserve the group action to online SetDisjointness. The Abelian 3-Orthogonal Array (3OA) problem is the problem of deciding if a set or array $S$ with elements drawn from a known finite Abelian group $(G,+)$ contains three distinct elements $a,b,c$ such that $a+b=c$ (we call such a triple a 3OA witness). Thus the 3SUM problem can be defined as the 3OA problem on the group $(\mathbb{Z}_{2n^3}, +)$. Like 3SUM, 3OA can be solved in $\tilde{O}(n)$ time by an unstructured search over pairs\footnote{The Fundamental Theorem of Finite Abelian Groups guarantees that there is always an efficient algorithm for the group operation.}.

In order to show that the time complexity of 3OA depends on the time complexity of SetDisjointness we need to show a reduction that runs in sub-linear time. To do this, we first use the $L$-subset finding quantum walk to reduce the given instance of 3OA to a sub-linear number of sub-linear size instances of 3OA with the data structures we need to reduce 3OA to SetDisjointness. 

Next we must show that, given those data structures, we can reduce 3OA to SetDisjointness in $o(n)$ time where $n$ is the size of the smaller instances in the first step. The main impediment to this is that the SetDisjointness instances we use rely on the decisions of a super-linear number of queries. 
For online SetDisjointness, we use variable-time amplitude amplification to reduce the total number of queries we need to make. Using variable-time amplitude amplification instead of Grover search allows us to add the step of searching for the 3OA witness at minimal cost, which means that our reduction only has one-sided error even if the SetDisjointness oracle has two-sided error. 

Next we will give an overview of how an algorithm for SetDisjointness can be used to solve 3OA.

\subsubsection{Constructing a SetDisjointness instance}

Recall that we can solve 3OA by searching over all $a \in S$ and $b \in S$ such that $a+b \in S$. The first idea for reducing 3OA is to group the first and third inclusions into buckets $\{ \mathcal{B}_i\}_{i \in R}$ with labels from a set $R$ with size $|R| = n^\gamma$ controlled by a parameter $\gamma \in (0, 1)$, so that we are searching for a triple $b, i, j \in S \times R^2$ such that there is some $a \in \mathcal{B}_i$ where $a+b \in \mathcal{B}_j$. For each triple $(b,i,j)$ this is equivalent to asking if the sets $\{a + b | a \in \mathcal{B}_i\}$ and $\mathcal{B}_j$ are disjoint or not, achieving our goal of reducing 3OA to SetDisjointness. However, this reduction is quite inefficient: the SetDisjointness instance has size $N=O(n^2)$ since there is a set for each bucket plus $|S|=n$ additional sets per bucket for a total of $n|R|+|R|$ sets of size $n/|R|$ each, $q=n|R|^2$ disjointness queries (one for each triple), and the universe size is $|G|$. 

We can apply three improvements to the reduction with the use of linear family of hash functions. A hash function $h$ is linear for a group $(G,+)$ if it preserves the group operation, meaning that $h(a) + h(b) = h(a+b)$. For the purposes of this overview we will assume that we have access to perfectly linear hash functions which have any additional properties we find convenient. 

The first improvement we can make with a family of such functions from $G$ to $R$ is to group all the input elements with the same hash value into the same bucket. This means that, for any $b \in S$ and $i \in R$, any element $a + b$ where $a \in \mathcal{B}_i$ must be in $\mathcal{B}_{i+h(b)}$, so we can specify a query using only a pair $b,i \in S \times R$ instead of a triple which reduces the number of queries to $n |R|$.

The second improvement is to use a linear hash family to compress the universe down to $O(n^2 / R^2)$ many elements. This increases the probability of a hash collision causing a false intersection to a constant, but we can bring that down to $O(1/\tmop{poly}(n))$ sampling $O(\log n)$ many functions from the family and counting two sets as disjoint if they are disjoint under any single hash function. This increases the number of sets and queries by an $O(\log n)$ factor. 

The third improvement is to split the $+b$ operation across both $\mathcal{B}_i$ and $\mathcal{B}_j$ using what we call a ``P{\u a}trașcu Split'' of the universe. Since the size of a P{\u a}trașcu Split depends on the size of the universe it benefits greatly from universe compression so we will assume that the second improvement has already been performed and use $h$ to represent the compression hash and $Q$ to represent the compressed universe. Informally, the idea is that if any $h(b) \in Q$ can be split into $h(b)^{\uparrow}$ and $h(b)^{\downarrow}$ such that $h(b) = h(b)^{\uparrow} + h(b)^{\downarrow}$ then $a + h(b)^{\uparrow} = a + h(b) - h(b)^{\downarrow}$ so checking if $\{a + h(b)^{\uparrow} | a \in \mathcal{B}_i\}$ and $\{c - h(b)^{\downarrow} | c \in \mathcal{B}_j\}$ are disjoint is equivalent to checking if $\{a + h(b) | a \in \mathcal{B}_i\}$ and $\mathcal{B}_j$ are disjoint. For example, taking $h(b)^{\uparrow}$ and $h(b)^{\downarrow}$ to be the upper and lower half of the bits of $h(b)$ is a P{\u a}trașcu split that works for integer addition and bit-string xor. If there are only $\sqrt{|Q|}$ possible upper elements and $\sqrt{|Q|}$ possible lower elements then we only need $(n/R) R$ ``left'' sets and $(n/R) R$ ``right'' sets instead of the $nR$ sets ``left`` sets we needed before.

For any $\gamma \in (0,1)$, combining all three improvements gives a reduction from 3OA with a data structure that separates $S$ into buckets and supports checking whether any given element is in $S$ or a bucket to a SetDisjointness instance with $2n\log n$ sets each with size $O(n^{1-\gamma})$ for an instance size of $N=O(n^{2-\gamma} \log n)$, $q=n^{1+\gamma} \log n$ queries, and a universe with size $O(n^{2-2\gamma})$.

\subsection{Conjectures and open questions}

Since there are no algorithms for the 3OA problem for any specific group which significantly improve upon the generic search-based algorithm that works for 3OA on any Abelian group, and since we know that differences between groups is irrelevant for the quantum query complexity of 3OA, it is natural to conjecture that the differences between groups also do not matter for the time complexity of 3OA. 

\begin{conjecture}[3OA Conjecture]
    The 3OA problem requires $\Omega(n^{d-o(1)})$ time with the same constant $d$ for all Abelian groups. Furthermore, $d=1$ (or $d=2$ in the classical setting).
\end{conjecture}
We also conjecture that the same situation applies to Abelian $k$-Orthogonal Array problems for all $k$.

\begin{conjecture}[kOA Conjecture]
    The Abelian $k$-Orthogonal Array (kOA) problem requires $\Omega(n^{\lceil k / 2 \rceil / 2 - o(1)})$ time ($\Omega(n^{\lceil k / 2 \rceil - o(1)})$ classical time) for any Abelian group.
\end{conjecture}

\section{Preliminaries}

\subsection{Model}

We use the standard quantum circuit model with random access gates as our model of quantum computation. In this model circuits always start in the all-zeros state, input is given in the form of special oracle gates that can be queried in superposition and act as $\mathcal{O}\ket{i}\ket{r} = \ket{i}\ket{r+S_i \mod R}$ (where $R$ is the maximum integer for the register and $S$ is the array of input integers), and the output is the state of a designated register after the end of the circuit. The random access gate swaps single bit in with another bit at a certain index in a register which is formalized by the mapping $\ket{i,b,r} \rightarrow \ket{i, r_i, r_1 \ldots b \ldots}$; the random access gates are required for data-structure operations and makes the model log equivalent to the QRAM model.

For our model of classical computation we use the following restriction of our quantum model: we only allow the X and CNOT gates in our elementary gate set, which also naturally restricts the ability to query oracles in superposition. Since the random access gate is still available this model is log equivalent to the Word RAM model.

The time complexity of a circuit is the number of elementary gates it has and the space complexity is the number of wires it needs. In this paper we will mostly ignore log factors 
in our complexity analysis. If we forget to specify which model we are referring to we are 
probably referring to the quantum model. 

\subsection{A brief explanation of the quantum subroutines used in this paper}
In this paper we use Grover search \cite{Gro96}, variable-time amplitude amplification \cite{Amb12}, and the $L$-subset finding quantum walk \cite{CE05,Amb04} as subroutines. The way use these algorithms work when we use them is that they first construct a uniform superposition over a range of possible states, then, using a routine we provide to recognize desirable states, they boost the amplitude of desirable states at the expense of undesirable ones, and finally they measure the superposition to obtain a single state. \emph{If} one or more desirable states exist, then there is a high probability that one of the desirable states is obtained after measurement. Usually the routine we provided for recognizing if a state is desirable is used on the measured state to produce the final search or decision output, but we can perform some other computation on the measured state if we choose.

We assume that the reader is familiar with the time complexity of Grover search, which runs in $\tilde{O}(\sqrt{n} \cdot t_m(n))$ time where $n$ is the number of items being searched and $t_m(n)$ is the time complexity of the routine provided to decided if an item is the one we are looking for.

\subsection{Definitions}

In this subsection we give formal definitions for the problems, properties, and other terms we use in the rest of this paper. We start with the 3OA problem, which can be thought of as the 3-Orthogonal Array Problem in \cite{BS13, Spa13a} restricted to Abelian groups. 

\begin{definition}[Abelian 3-Orthogonal Array ``3OA'' Problem]
  ~
  
  Input: an index oracle for $S$, an array of elements from an Abelian group
  $(G, +)$, with instance size $n = | S |$.
  
  Decision Output: YES if and only if $\exists x, y \in S$ s.t. $(x + y) \in
  S$.
  
  Search Output: $x, y \in S$ s.t. $(x + y) \in S$.
\end{definition}
It is known that size-$n$ 3OA instances on groups of arbitrary size can be reduced to instances on groups of $O(n^3)$ size via hashing if the group has (almost) linear hashes; see the appendix of \cite{JV16} for an example of this. Next we give definitions for two well-studied examples of the 3OA problem. 
\begin{definition}[3SUM Problem]
  3OA with group $(\mathbb{Z}_{2 n^3}, +_{2 n^3})$.
\end{definition}
It is known that this definition of 3SUM is equivalent to $(\mathbb{Z}, +)$ for the reason above.
\begin{definition}[\cite{JV16} 3XOR Problem]
  3OA with group $(\mathbb{Z}_2^{3 \log n}, \oplus)$ where $\oplus$ denotes
  element-wise mod-2 addition.
\end{definition}

Next we define the convolution variant of the 3OA problem. 

\begin{definition}[Convolution 3OA ``Conv3OA'' Problem]
  ~
  
  Input: an index oracle for $S$, an array of elements from an Abelian group
  $(G, +_G)$, with instance size $n = | S |$.
  
  Decision Output: YES if and only if there exist distinct 
  $i, j \in [|S|]$ s.t. $S_i +_G S_j = S_{i+j}$.
  
  Search Output:  $i, j \in [|S|]$ s.t. $S_i +_G S_j = S_{i+j}$.
\end{definition}
Note that we use integer addition for the indices.

Next we switch from defining problems to defining the terms and properties used in our reduction framework, starting with the hash function properties needed for the improvements we mentioned in the techniques section.
\begin{definition}[$\alpha$-almost balanced hash family]
  A family of hash functions $h_i : G \rightarrow R$ is $\alpha$-balanced if,
  for some fixed constant $c''$, for any $S \subseteq G$ the expected number
  of elements $x \in S$ such that $| \{ y = h (x) |y \in S \} | > c''  \frac{|
  S |}{| R |}$ is $\alpha$. We call $c''$ the size factor of the family. 
\end{definition}

\begin{definition}[$\alpha$-almost linear hash function]
  A function $h : G \rightarrow R^{\alpha}$ is almost linear for the group
  $(G, +_G)$ if $(R, +_R)$ is a quasigroup and for all $a, b \in G$ there
  exists some $i \in [\alpha]$ such that $h (a)_0 +_R h (b)_0 = h (a +_G
  b)_i$.
\end{definition}
When we say that a hash function is ``almost linear'' while omitting 
the parameter $\alpha$ we mean that the function is $\tilde{O}(1)$-almost linear. 
The next property is one we need for the \emph{range} of the hash family we 
use to compress the universe of SetDisjointness elements.

\begin{definition}[P{\u a}trașcu Split]
  A P{\u a}trașcu Split of a quasi-group $(Q, +)$ is a pair sets $Q^{\uparrow}$ and
  $Q^{\downarrow}$ with the following three operations: a bijection between
  $Q$ and $Q^{\uparrow} \times Q^{\downarrow}$, \ $\pm^{\uparrow} : Q \times
  Q^{\uparrow} \rightarrow Q$, and $\pm^{\downarrow} : Q \times Q^{\downarrow}
  \rightarrow Q$. We invoke the bijection using arrow superscripts like $c
  \leftrightarrow (c^{\uparrow}, c^{\downarrow})$. The operations
  $\pm^{\uparrow}$ and $\pm^{\downarrow}$ have the property that for $a, b, c
  \in Q$, $a \pm^{\uparrow} b^{\uparrow} = c \pm^{\downarrow} b^{\downarrow}$
  if and only if $a + b = c$.
  
  The size of the split is $\max (| Q^{\uparrow} |, | Q^{\downarrow} |)$.
\end{definition}

For convenience, we give a name to the data structure we need to reduce 3OA to 
SetDisjointness in sub-linear time.
\begin{definition}[Hash-bucket data structure]
  A hash-bucket data structure for a hash function $h_1$ is a data structure
  that stores the elements from some set $S$ and the sets $\mathcal{B}_i
  \assign \{ x \in S|h_1 (x) = i \}$ which we call buckets which support the
  following operations in $\op{\tilde{O} (| S |)}$ time:
  \begin{enumerateroman}
    \item Access to the elements of each $\mathcal{B}_i$ by index (and access
    to each $| \mathcal{B}_i |$)
    
    \item Membership oracles for each $\mathcal{B}_i$
    
    \item Access to the elements of $S$ by index (and access to $| S |$)
    
    \item A membership oracle for $S$.
  \end{enumerateroman}
\end{definition}

\subsection{Conjectures}

Below we reproduce the widely-conjectured hardness of the 3SUM problem in the quantum and classical settings. In combination with our reductions, they can be used to provide conditional lower bounds on other problems.

\begin{conjecture}[\cite{AL20} Quantum 3SUM conjecture]
    The 3SUM problem cannot be solved in $O(n^{1-\delta})$ quantum time for any constant $\delta > 0$.
\end{conjecture}

\begin{conjecture}[\cite{Vas15} Classical 3SUM conjecture (modern version)]
    The 3SUM problem cannot be solved in $O(n^{2-\delta})$ classical time for any constant $\delta > 0$.
\end{conjecture}

\section{Our reduction framework}

In this section we describe our framework for reducing 3OA to
online SetDisjointness.

\subsection{Data structure and time complexity lemmas}

Our reductions require organizing the input 3OA elements into a hash-bucket data structure to work and they must run in sub-linear time, but placing all the input elements into a data structure requires $\tilde{O}(1)$ time. To solve this problem, our reductions use a quantum walk with a history-independent data structure to reduce the input 3OA instance with no additional structure to a sub-linear number of sub-linear sized 3OA instances with inputs that are organized in a hash-bucket data structure, which we can then solve by our ``inner'' reduction to SetDisjointness.

We start by showing that there is a history-independent ordered set which supports the operations we need to implement our hash-bucket data structure.

\begin{lemma}[Existence of a history-independent ordered set]
  \label{lem:skip-data-structure}
  
  There is an history independent data structure $D$ which, when given
  \begin{itemizedot}
    \item a comparison function for some universe $U$,
    
    \item a unique number with $\Omega (\log n)$ bits for each element of
    $U$,
    
    \item a promise that at most $n$ unique elements of $U$ will be stored in
    it,
    
    \item a promise that at most $n^q$ elements will be stored in it at once
    with $q \in [0, 1)$,
    
    \item and a constant $c_1$, 
  \end{itemizedot}
  allows the following operations in $O (\log^4 n)$ time:
  \begin{enumeratenumeric}
    \item Inserting an element into $D$
    
    \item Removing an element from $D$
    
    \item retrieving the $i^{\tmop{th}}$ smallest element in $D$ according to
    the ordering given by the comparison function
    
    \item Checking if a given element is in $D$.
  \end{enumeratenumeric}
  The probability that any one operation fails is $O \left( \frac{1}{n^{c_1 -
  1}} \right)$, and the total space used by the data structure is $\tilde{O}
  (n^q)$.
\end{lemma}

\begin{proof}
  \cite{BLPS22}, using and extending \cite{Amb04}, show most of this statement 
  implicitly, with
  the sole exception that they give the error probability for one operation as
  $O \left( \frac{\log n}{n^4} \right)$. We will now remedy this with a
  tighter analysis.
  
  The data structure in question consists of a hash table and skip list.
  There are two ways that an operation can fail: first one of the buckets in
  the hash table can run out of entries, second accessing an entry of the skip
  list could take too much time.
  
  The probability of the first failure is given in the proof of lemma 6 of
  \cite{Amb04} as $\frac{e^{\log n - s - 1}}{(\log n - s)^{\log n - s}}$, where
  $s$ is some constant. \cite{Amb04} simply states that this is in $o (1 / n^4)$,
  but a tighter analysis shows that this is
  \[ O \left( \left( \frac{e}{\log n} \right)^{\log n} \right) = O \left(
     \left( \frac{1}{2^{\log \log n}} \right)^{\log n} \right) = O \left(
     \frac{1}{n^{\log \log n}} \right), \]
  which is $o \left( \frac{1}{n^{c_1}} \right)$ for any constant $c_1$.
  
  The probability of the second failure for a fixed index or element in $D$ is
  given in Lemma 3 of \cite{BLPS22} and Lemma 6 of \cite{Amb04} as $O \left(
  \frac{\log n}{2^d} \right)$, where $d = c_1 \log n + 1$ for some constant
  $c_1$; this comes from the fact that the skip list uses hash functions from
  a $d$-wise indepent family to decide the number of levels that a node is
  stored in. They choose the specific value $c_1 = 4$, however it is known
  that any constant value can be used (see Theorem 1 of \cite{Amb04} or Footnote
  9 of \cite{BLPS22}). By substituting the general value of $d$ we get the
  following tighter bound: $O \left( \frac{\log n}{n^{c_1}} \right)$. Since
  there are at most $n^q$ indices or elements in $d$ and $q < 1$, by the union
  bound, the probability of the second type of failure without fixing the
  index or element involved is at most
  \[ O \left( n \cdot \frac{1}{n^{c_1}} \right) = O \left( \frac{1}{n^{c_1 -
     1}} \right) \]
  By the union bound the probability that one or both failures occur in a
  single operation is given by $O \left( \frac{1}{n^{c_1 - 1}} \right)$ (since
  the probability for the first type of failure is asymptotically smaller).
\end{proof}

Next we find out how to calculate the probability that our reductions are 
affected by a single error in an operation on the ordered set. 

\begin{lemma} \label{lem:skip-data-structure-cumulative-error}
  The distance between the quantum state representing the data structure of
  \Cref{lem:skip-data-structure} after
  $t$ perfect operations and the quantum state after $t$ operations with error
  probability $O (\varepsilon)$ is $O \left( t \sqrt{\varepsilon} \right)$.
\end{lemma}

\begin{proof}
  Given by substituting the general probability into the relevant parts of the
  proofs for Lemmas 1 and 2 in \cite{BLPS22}.
\end{proof}

We finish with a lemma that allows us to derive the overall time complexity 
and success probability of our reduction using what we know about our data 
structure and our inner reduction. 

\begin{lemma}[Time complexity of the $L$-subset finding walk]
\label{lem:time-of-johnson-walk}
  Suppose that:
  \begin{enumerate}
    \item When given access to the history-independent data structure $D [V]$
    on $V$, there is a quantum algorithm which solves an L-Subset Finding
    Problem with probability $1 - \frac{1}{\tmop{poly} (n)}$ on a set $V$ with
    $M$ elements in time $T_D (M)$.
    
    \item $D [V]$ can be constructed in time $T_S (M)$.
    
    \item $D [V]$ can be updated to add or remove one element in $V$ in $T_U
    (M)$ time.
  \end{enumerate}
  Further suppose that $T_D (M) = M^{1 - \delta}$ for some $\delta \in
  \mathbb{R}$, $T_S (M) = \tilde{O} (M)$, and $T_U (M) = \tilde{O} (1)$.
  
  Then there is a quantum algorithm that solves the same L-Subset Finding
  Problem on a set $S$ with $N$ elements with probability $1 - o (1)$ in
  sub-linear time given access only to an index oracle for $S$. More
  specifically, when $\delta < \frac{1}{2}$ there is an algorithm with time
  $\tilde{O} \left( {N^{\frac{L}{L + 2 \delta}}}  \right)$; otherwise there is
  an algorithm with time $\tilde{O} \left( {N^{\frac{L}{L + 1}}}  \right)$.
\end{lemma}

\begin{proof}
  \cite{CE05} show that a quantum algorithm $(W^{t_1} P)^{t_2}$ applied to the
  state $| s \rangle$ solves an $L$-Subset Finding problem on $S$, where $P$
  is a unitary which flips the phase of the state $| V, D [V] \rangle$ if and
  only if $S$ contains a valid $L$-subset, $W$ is a walk unitary that can be
  constructed using a unitary $U_{\tmop{update}}$ which takes $| D [S], x
  \rangle$ to $| D [S'], x \rangle$ where $S' = S \cup \{ x \}$ or $S' = S
  \setminus \{ x \}$ depending on the size of $S$, and $| s \rangle$ is the
  uniform superposition over the states $| V, D [V] \rangle \sum_{x \in S : x
  \nin V} | x \rangle$ where $V$ is a suitable representation of a size $M =
  N^q$ subset of $S$ with $q \in (0, 1)$. With $t_1 = \left\lfloor
  \frac{\pi}{2} \sqrt{M / L} \right\rceil$ and $t_2 = \left\lfloor
  \frac{\pi}{2} (N / M)^{L / 2} \right\rceil$, the probability that this
  algorithm fails is $O (1 / M + M / N) = O (1 / \tmop{poly} (N))$.
  
  It is straightforward to implement the unitary $P$ in O ($T_D (M)$)  time
  using the given data structure dependent quantum algorithm, the unitary
  $U_{\tmop{update}}$ in $O (T_U (M))$ time using the data structure's remove
  and add operations once each, and the unitary $| V, 0 \rangle \mapsto | V, D
  [V] \rangle$ in $O (T_S (M))$ time using the data structure's constrution
  operation. The remaining parts of $W$ can be implemented in $\tilde{O} (1)$
  time and the remaining parts of $| s \rangle$ can be computed in $\tilde{O}
  (M)$ time, so the total time complexity of the algorithm is $\tilde{O} (M +
  T_S (M) + t_1 t_2 T_U (M) + t_2 T_D (M))$, which is
  \[ \tilde{O} \left( M + T_S (M) + \sqrt{M}  \sqrt{(N / M)^L} T_U (M) +
     \sqrt{(N / M)^L} T_D (M) \right) \]
  . Rewritting in terms of $N$ gives $\tilde{O} (N^q + N^{L / 2 + q (1 - L) /
  2} + N^{L / 2 + q (1 - \delta - L / 2)})$. This is clearly sub-linear for
  all $1 > q > \max \left( \frac{L - 2}{L - 1}, \frac{L - 2}{L - 2 + 2 \delta}
  \right)$. If $\delta < \frac{1}{2}$ then we can pick $q = \frac{L}{L + 2
  \delta}$ which gives the desired bound of $\tilde{O} \left( N^{\frac{L}{L +
  2 \delta}} + N^{\left( \frac{L}{(L + 2 \delta)}  \left( \frac{1}{2} + \delta
  \right) \right)} {+ N^{\frac{L}{L + 2 \delta}}}  \right) = \tilde{O} \left(
  {N^{\frac{L}{L + 2 \delta}}}  \right)$. If $\delta \geqslant \frac{1}{2}$
  then we can pick $q = \frac{L}{L + 1}$ which gives the desired bound of
  $\tilde{O} \left( N^{\frac{L}{L + 1}} + N^{\frac{L}{L + 1}} + N^{\frac{L}{L
  + 1}  \left( \frac{3}{2} - \delta \right)} \right) = \tilde{O} \left(
  N^{\frac{L}{L + 1}} \right)$.
  
  The probability that the algorithm does not fail due to a failure of the
  given quantum algorithm is lower bounded by $\left( 1 - \frac{1}{\tmop{poly}
  (n)} \right)^{t_2}$, so the overall probability of success is
  \[ \left( 1 - \frac{1}{\tmop{poly} (n)} \right)^{t_2} \left( 1 -
     \frac{1}{\tmop{poly} (n)} \right) = 1 - o (1) . \]
\end{proof}

Note that in the final condition gives an algorithm that is tight for the
query lower bound given by \cite{BS13}. Note also that our analysis works for all
$\delta$, including $\delta = 0$ and negative $\delta$; however for such
values of $\delta$ the quantum walk does not give a time complexity
improvement for the problems in this paper over building the data structure by
iterating over the input elements and then running the algorithm that relies
on the data structure.

The following corollary immediately follows from the lemma:

\begin{corollary} \label{cor:sublinear-walk}
  If an $L$-subset finding problem can be solved in sub-linear time (quantum or classical)
  using a data structure with $\tilde{O}(1)$-time add and remove operations,
  then there is a sub-linear-time quantum algorithm that solves that $L$-subset 
  finding problem without additional structure. 
\end{corollary}

\subsection{Framework for online SetDisjointness}

Now we are ready to describe our reduction framework for online SetDisjointness. 
We start with our reduction from 3OA with a hash-bucket data structure for a fixed 
value of $\gamma$ to online SetDisjointness. This will be the inner reduction 
used by the quantum walk.

\begin{lemma}
  \label{lem:red-with-ds-to-online}
  Let $(G, +)$ be an abelian group and $\gamma \in (0, 1)$ be a parameter.
  Suppose there is a family of hash functions $h_1 : G \rightarrow R$ that is
  $n^{\beta}$-almost balanced and $\tilde{O} (1)$-almost linear for $(G, +)$
  and a family of hash functions $h_2 : G \rightarrow Q$ that is pairwise
  independent and $\tilde{O} (1)$-almost linear for $(G, +)$ where $| R | =
  n^{\gamma}$, $| Q | = c' n^{2 - 2 \gamma}$ for sufficiently large $c' \in
  \tilde{O} (1)$, $Q$ has a P{\u a}trașcu split of size $O \left( \sqrt{| Q
  |} \right)$, and $h_1$ and $h_2$ can be computed in $\tilde{O} (1)$ quantum
  time.
  
  Then there is a reduction from 3OA for $(G, +)$ with a hash-bucket data
  structure for $h_1$ to Online SetDisjointness with bounded error on
  $\tilde{O} (n)$ sets with $O (n^{1 - \gamma})$ elements each from a universe
  with $O (n^{2 - 2 \gamma})$ elements, where $n$ is the size of the 3OA
  instance. The reduction solves 3OA for $(G, +)$ in $\tilde{O} (n^{(1 +
  \beta) / 2} + t_p (N) + n^{(1 + \gamma) / 2} t_q (N))$ time where $N =
  \tilde{O} (n^{2 - \gamma})$ is the size of the Online SetDisjointness
  instance and $t_p (N)$ and $t_q (N)$ are respectively the preprocesing and
  query complexities for an Online SetDisjointness algorithm. The reduction
  produces false negatives with $\frac{1}{\tmop{poly} (n)}$ probability and no
  false positives.
\end{lemma}

\begin{proof}
  Let $R'$ be the set of labels for the balanced buckets. 
  Let $\eta$ be a parameter to be set
  later. Sample $\eta$ hash functions $h_{2, 0}, \ldots, h_{2, \eta - 1}$ from
  the $h_2$ family. Without loss of generality let the size of the P{\u a}trașcu
  split be $| Q^{\uparrow} |$.
  
  We now describe our reduction from a size-$n$ instance $S$ of 3OA for $(G,
  +)$ to Online SetDisjointness. Remember that we have a data structure
  representing the buckets $\mathcal{B}_i \assign \{ a \in S|h_1 (a) = i \}$
  where $i \in R$ and thus the balanced buckets $R'$ with access to efficient
  oracles for $| \mathcal{B}_i |$, indexing into $\mathcal{B}_i$, membership
  in $\mathcal{B}_i$, and memebership in $S$. Our reduction must efficiently
  implement oracles for the number of sets, the size of each set, and access
  to set elements by index of our SetDisjointness instance.
  
  We produce an instance of Online SetDisjointness with $2 | R' |  |
  Q^{\uparrow} | \eta$ sets in $\mathcal{F}$. Each set corresponds to a
  quadruple $(t, i, j, l) \in \{ \uparrow, \downarrow \} \times R' \times Q^t
  \times [\eta]$. The quadruple $(t, i, j, l)$ can always be compactly
  represented by a single index by fixing some orderings on $R'$,
  $Q^{\uparrow}$, and $Q^{\downarrow}$ and retrieving $i$ and $j$ using those
  orderings. The ordering on $R'$ can be created in $O (| R |)$ time and the
  orderings on $Q^{\uparrow}$ and $Q^{\downarrow}$ can be given as a part of
  the description of the P{\u a}trașcu split of $Q$. We implement the oracles
  for Online SetDisjointness as follows:
  \begin{eqnarray*}
    | \mathcal{F} | & \assign & 2 | R' |  | Q^{\uparrow} | \eta\\
    | \mathcal{F}_{(t, i, j, l)} | & \assign & | \mathcal{B}_i |\\
    (\mathcal{F}_{(t, i, j, l)})_k & \assign & h_{2, l} ((\mathcal{B}_i)_k)
    \pm^t j.
  \end{eqnarray*}
  We use this Set Disjontness instance to decide if there is a pair $(a, c)$
  in the balanced buckets such that \ $a + b = c$ for some $b \in S$. We do
  this by looking for a marked pair $(i, b) \in R' \times S$ using
  variable-time amplitude amplification with the following two-step marking
  procedure:
  \begin{enumeratenumeric}
    \item Use linear search to look for an $l \in [\eta]$ such that
    $\mathcal{F}_{(\uparrow, i, h_{2, l}^{\uparrow} (b), l)}$ and
    $\mathcal{F}_{(\downarrow, i + h_1 (b), h_{2, l}^{\downarrow} (b), l)}$
    are disjoint. If such an $l$ is found, terminate without marking.
    
    \item Use Grover to search for $a \in \mathcal{B}_i$ such that $a + b \in
    \mathcal{B}_{i + h_1 (b)}$ using the bucket membership oracle. If found,
    mark $(i, b)$.
  \end{enumeratenumeric}
  The second step of the marking procedure, and thus the use of variable-time
  amplitude amplification, is not strictly necessary but it allows us to
  eliminate false positives for a negligible time complexity cost. If a marked
  pair is found then there exists a 3OA witness; if desired we can find the
  3OA witness by performing step 2 again.
  
  For the sake of readability we have presented this reduction assuming that
  $h_1$ and $h_2$ are perfectly linear. Since they are both $\tilde{O}
  (1)$-almost linear, we can ensure that every pair $(i, b)$ that would be
  marked by our marking procedure using perfectly linear hash functions is
  also marked when using the almost linear hash functions by increasing the
  number of sets, potential set queries in step 1, and the number of buckets
  to query in step 2 by an $\tilde{O} (1)$ factor. Doing so increases the
  number of false positives (when step 2 of the marking procedure is run but
  the pair is not marked) by at most a $\tilde{O} (1)$ factor.
  
  It remains to handle the case where one of the elements of the 3OA witness
  is in one of the heavy buckets $R - R'$ and the identification of balanced
  and heavy buckets. Since we have efficient access to the size of the
  buckets, we can identify the buckets and construct the tables needed to
  efficiently index in to balanced and heavy buckets by iterating over all the
  buckets. We can find a 3OA certificate with an element in the heavy buckets
  by Grover search over all pairs $(a, b)$ where $a$ is from a heavy bucket
  and $b \in S$ for a pair such that $a + b \in S$ using the membership oracle
  for $S$.
  
  The size of the SetDisjointness instance is
  \[ N = | \mathcal{F} |  | \mathcal{F}_{(t, i, j, l)} | = \tilde{O} \left(
     n^{\gamma}  \sqrt{Q} n^{1 - \gamma} \right) = \tilde{O} (n^{2 - \gamma})
     . \]
  This concludes the description of our reduction. Next we will analyze the
  error probabilities of our reduction.
  
  \paragraph{Error Probability}
  
  When an instance has no 3OA witness the reduction will never report one
  since a 3OA witness is required to mark a pair in step 2 of the marking
  procedure (remember that every member of a bucket is also a member of $S$).
  
  When an instance has a 3OA witness $(a, b, a + b)$ then the pair $(i, b) =
  (h_1 (a), b)$ will pass step 1 with perfect SetDisjointness queries because
  $h_{2, l} (a) \pm^{\uparrow} h_{2, l}^{\uparrow} (b) = h_{2, l} (a + b)
  \pm^{\downarrow} h_{2, l}^{\downarrow} (b)$ for all $l$ by the linearity of
  $h_2$ and the properties of the P{\u a}trașcu Split, and $h_{2, l} (a)
  \pm^{\uparrow} h_{2, l}^{\uparrow} (b) \in \mathcal{F}_{(\uparrow, i, h_{2,
  l}^{\uparrow} (b), l)}$ and $h_{2, l} (a + b) \pm^{\downarrow} h_{2,
  l}^{\downarrow} (b) \in \mathcal{F}_{(\downarrow, i + h_1 (b), h_{2,
  l}^{\downarrow} (b), l)}$ by the linearity of $h_1$. If each Set
  Disjointness query returns a true positive with a probability bounded away
  from $\frac{1}{2}$, then the probability of passing step 1 can be $1 -
  \frac{1}{\tmop{poly} (n)}$ by taking the majority of $\tilde{O} (1)$
  repetitions for each $l$. Grover search finds $a$ in step 2 with high
  probability (w.h.p.) and variable-time amplitude amplification finds a
  marked $(i, b)$ w.h.p. which can be boosted to $1 - \frac{1}{\tmop{poly}
  (n)}$ with $\tilde{O} (1)$ repetitions. The probability that the 3OA
  witness is found is the probability that $(i, b)$ passes step 1, and is
  marked in step 2, and is found by VTAA, and then Grover search over
  $\mathcal{B}_i$ finds $a$ again, which is $\left( 1 - \frac{1}{\tmop{poly}
  (n)} \right)^4 = 1 - \frac{1}{\tmop{poly} (n)}$ as required. If successive
  SetDisjointness queries do not give right or wrong answers independently
  but instead have a fixed answer for each instance created by the
  proprocessing step then we can achieve the same error bounds by $\tilde{O}
  (1)$ repetitions of the whole reduction.
  
  \paragraph{Time Complexity}Our reduction consists of five sequential steps:
  \begin{enumerateroman}
    \item Creating tables to index into balanced and heavy buckets: $\tilde{O}
    (| R |) = O (n^{\gamma})$.
    
    \item Handling heavy buckets: $\tilde{O} \left( \sqrt{n^{\beta}  | S |}
    \right)$=$\tilde{O} (n^{(1 + \beta) / 2})$.
    
    \item Creating the SetDisjointness oracle: $\tilde{O} (1)$.
    
    \item Online SetDisjointness preprocessing: $\tilde{O} (t_p (N))$.
    
    \item The variable-time amplitude amplification procedure.
  \end{enumerateroman}
  We will now work out the time complexity of the variable-time amplitude
  amplification procedure. The time complexity of variable-time amplitude
  amplification \cite{Amb12} is
  \[ O \left( T_{\max}  \sqrt{\log T_{\max}} +
     \frac{T_{\tmop{av}}}{\sqrt{p_{\tmop{mark}}}} \log^{1.5} T_{\max} \right)
  \]
  where $T_{\max}$ is the worst-case time complexity of the marking procedure,
  $T_{\tmop{av}}$ is the $\ell_2$ average of the time complexities of the
  marking procedure, and $p_{\tmop{mark}}$ is the fraction of pairs that are
  marked. Let $T_1$ and $T_2$ be the time complexities of the first and second
  step of the marking procedure respectively; then $T_1 = \tilde{O} (\eta T_q
  (N))$ and $T_2 = \tilde{O} \left( \sqrt{n / R} \right)$ since
  $\mathcal{B}_i$ is gauranteed to be balanced and thus contains $O (n / R)$
  elements. We can then write $T_{\max}$ and $T_{\tmop{av}}$ in terms of $T_1$
  and $T_2$ as $T_{\max} = T_1 + T_2$ and $T_{\tmop{av}} =
  \sqrt{p_{\tmop{int}}  (T_{\max})^2 + (1 - p_{\tmop{int}}) T_1^2}$ where
  $p_{\tmop{int}}$ is the probability that the marking procedure does not
  terminate after step 1.
  
  There are two ways for the marking procedure to proceed to step 2: first
  there could be no disjoint sets because there is a 3OA witness, second
  there could be no disjoint sets due to collisions or erroneous query results
  for all $\eta$ hash functions despite the lack of a 3OA witness. The
  probability of the first case is $p_{\tmop{mark}}$. In the worst case only
  one pair $(i, b)$ will be marked, so we will use $p_{\tmop{mark}} =
  \frac{1}{Rn} = \frac{1}{n^{1 + \gamma}}$. We now analyze the probability of
  the second case. For the $l^{\tmop{th}}$ item in step 1, the probability
  that there is an intersection between the queried sets is equivalent to the
  probability that $h_{2, l} (a) + h_{2, l} (b) = h_{2, l} (c)$ for some $a
  \in \mathcal{B}_i, c \in \mathcal{B}_{i + h_1 (b)}$ and $c \neq a + b$ which
  is equivalent to $h_{2, l} (a + b) = h_{2, l} (c)$ by linearity. $\Pr [h_{2,
  l} (x) = h_{2, l} (y)] = 1 / Q$ when $x \neq y$ since $h_2$ is pairwise
  independent and $| \mathcal{B}_i | < \frac{c'' n}{R}$ for some constant
  $c''$ since the buckets are balanced, so the probability of an unwanted
  intersection for a single $l$ is $(c'' n^{1 - \gamma})^2 {/ Q = c''}^2 / c'$
  by the union bound. If $\varepsilon$ is the bounded probability that a Set
  Disjointness query falsely reports an intersection, then by choosing $c' =
  \frac{3}{1 / 2 - \varepsilon} {c''}^2$ the overall probability that a single
  query does not report disjointness is bounded away from $1 / 2$ from below,
  so by choosing $\eta = \tilde{O} (1)$ the probability that no disjoint sets
  are reported is $\frac{1}{\tmop{poly} (n)}$. Since the two cases are
  mutually exclusive, we have $p_{\tmop{int}} = \frac{1}{n^{1 + \gamma}} +
  \frac{1}{\tmop{poly} (n)} = O \left( \frac{1}{n^{1 + \gamma}} \right)$.
  
  Now we can put together the complexity of step v, starting with
  $T_{\tmop{av}}$:
  \[ T_{\tmop{av}} = \sqrt{p_{\tmop{int}} T_2^2 + p_{\tmop{int}} \cdot 2 T_1
     T_2 + T_1^2} = \tilde{O} \left( \sqrt{\frac{n^{1 - \gamma}}{n^{1 +
     \gamma}} + \frac{n^{(1 - \gamma) / 2}}{n^{1 + \gamma}} T_1 + T_1^2}
     \right) = \tilde{O} (T_1) = \tilde{O} (T_q (N)) . \]
  Then overall time complexity of step v is
  \[ \tilde{O} \left( T_{\max} + \frac{T_{\tmop{av}}}{\sqrt{p_{\tmop{mark}}}}
     \right) = \tilde{O} \left( T_q (N) + \sqrt{n^{1 - \gamma}} + \sqrt{n^{1
     + \gamma}} T_q (N) \right) = \tilde{O} (n^{(1 + \gamma) / 2} T_q (N)) .
  \]
  Adding the complexities of all the steps together gives the reduction an
  overall a time complexity of
  \[ \tilde{O} (n^{(1 + \beta) / 2} + t_p (N) + n^{(1 + \gamma) / 2} T_q (N))
  \]
  as required.
\end{proof}

Next we give a complete reduction from 3OA to online SetDisjointness for a 
single value of $\gamma$ using the reduction above to solve the inner 3OA 
instances of the quantum walk. 

\begin{lemma}[Reduction framework for online SetDisjointness] \label{lem:red-to-online}
  ~
  
  Let $\gamma \in (0,1)$ be a parameter. If:
  
  - there is an algorithm for Online SetDisjointness with $t_p (N)$ quantum
  preprocessing time and $t_q (N)$ quantum time per query which correctly
  answers queries with bounded error $\varepsilon$ such that
  \[ t_p (N) + N^{\frac{1 + \gamma}{2 (2 - \gamma)}} T_q (N) = O \left(
     N^{\frac{1 - \delta}{2 - \gamma}} \right) \]
  for some constant $\delta > 0$, and
  
  - there is an efficient family of hash functions $h_1 : G \rightarrow R$
  with $| R | = n^{\gamma}$ that is $n^{\beta}$-almost balanced for size
  factor $c''$ and $\tilde{O} (1)$-almost linear for the abelian group $(G,
  +_G)$ in the operation $+_R$ for some $\beta \in (0, 1)$, and
  
  - there is an efficient family of hash functions $h_2 : G \rightarrow Q$
  that is pairwise independent and $\tilde{O} (1)$-almost linear for $(G,
  +_G)$ in the operation $+_Q$, where $| Q | = \frac{3}{1 / 2 - \varepsilon}
  {c''}^2 n^{2 - 2 \gamma}$ and $Q$ has an efficient P{\u a}trașcu Split of
  size $O \left( \sqrt{|Q|} \right)$,
  
  then there is an algorithm which solves the 3OA problem on $(G, +_G)$ with
  size $n$ in sub-linear quantum time with $o (1)$ one sided error.
  Specfically, there is an $\tilde{O} \left( n^{\max \left( \frac{3}{4},
  \frac{3}{3 + 2 \delta} \right)} \right)$-time algorithm.
\end{lemma}

\begin{proof}
  We are given an ordered set $S$ as our 3OA instance. 
  \Cref{lem:skip-data-structure} gives an
  ordered-set data structure that supports adding an element, removing an
  element, accessing an element by index, and checking if an element is in the
  list all in $\op{\tilde{O} (1)}$ time and is history-independent. We can use
  $R$ ordered-set data structures to represent each bucket and one to store
  the current subset of $S$ to implement the hash-bucket data structure for
  $h_1$. We store each element as an index-value pair $(i, S_i)$ and use the
  index as the unique number and to implement the comparison function. We will
  call any individual instance of the ordered-set data structure in our
  hash-bucket data structure a set.
  
  We will temporarily assume that operations on our data structure always
  succeed and apply it to \Cref{lem:time-of-johnson-walk} to get an algorithm
  for 3OA using the reduction of \Cref{lem:red-with-ds-to-online} to solve 3OA
  with our data structure using an algorithm for online SetDisjointness. 
  
  Since $\beta < 1$, \Cref{lem:red-with-ds-to-online} gives an algorithm for 3OA with
  our data structure in $\tilde{O} (t_p (N) + n^{(1 + \gamma) / 2} T_q (N))$
  time; this is equal to $O (n^{1 - \delta})$ time by applying the time
  complexity of the online SetDisjointness algorithm and $N = \tilde{O} (n^{2
  - \gamma})$ from the lemma. Applying this complexity for 3OA with our data
  structure to \Cref{lem:time-of-johnson-walk} then gives the stated time complexity.
  
  The final task is to verify that the probability that the reduction fails
  due to data structure errors is sufficiently low. Let $c_1$ be a constant to
  be set later. By \Cref{lem:skip-data-structure} the error probability for an 
  operation on one of the
  sets is $O \left( \frac{1}{n^{c_1 - 1}} \right)$; since we have $O (R)$ sets
  by a union bound the error probability for one operation on an arbitrary set
  is $O \left( \frac{n^{1 + \gamma}}{n^{c_1}} \right)$. Since the algorithm
  uses sub-linear time it uses fewer than $O (n)$ data-structure operations
  total, so by \Cref{lem:skip-data-structure-cumulative-error} the distance between the final states of the quantum walk
  algorithm with and without data-structure errors is $O (n^{1 + (1 + \gamma -
  c_1) / 2}) = O (n^{3 / 2 + \gamma / 2 - c_1 / 2})$. So by choosing any
  sufficiently large $c_1 \gg 4$ this distance becomes $O \left(
  \frac{1}{\tmop{poly} (n)} \right)$. Thus the probability of the measured
  state after the quantum walk algorithm having an error due to data structure
  errors is also $O \left( \frac{1}{\tmop{poly} (n)} \right)$. Combining this
  with the $o (1)$ error from the quantum walk algorithm, the $O \left(
  \frac{1}{\tmop{poly} (n)} \right)$ error from the algorithm for 3OA with
  structure, and a check of the resulting 3OA witness (which can be done
  using the indices) gives an overall $o (1)$ one sided error as required.
\end{proof}

Now that we have a framework for generating a reduction from 3OA to online 
SetDisjointness for any fixed parameter $\gamma$, we can combine it with 
the assumption that there are hash families for any $0 < \gamma < 1$ to 
obtain a framework for lower bounds on a range of query-preprocessing 
tradeoffs for online SetDisjointness conditioned on the complexity of 3OA.

\restateonlineframeworktheorem*

\begin{proof}
  We aim to show that if $p + 2 q < 1$ then there is an $O (n^{1 -
  \delta})$-time algorithm for the 3OA problem on the abelian group $(G,
  +_G)$. Assume that there is an algorithm for Online SetDisjointness with $p
  + 2 q < 1$; then by applying \Cref{lem:red-to-online} using the given 
  family of hash functions
  it is sufficient to show that for any valid values for $p$ and $q$ there is
  a value of $0 < \gamma < 1$ such that $N^p + N^{\frac{1 + \gamma}{2 (2 -
  \gamma)}} N^q = O \left( N^{\frac{1 - \delta}{2 - \gamma}} \right)$ for some
  constant $\delta>0$. Let $\delta' = \min \left( \frac{1}{2 - \gamma} - p,
  \frac{1}{2 - \gamma} - \frac{1 + \gamma}{2 (2 - \gamma)} - q \right)$; then
  $\delta' = \delta / (2-\gamma)$ so 
  it is sufficient to show that $\delta' > 0$.
  
  We choose $\gamma = \frac{4 (p - q) - 1}{2 (p - q) + 1}$. Then $\frac{1 +
  \gamma}{2 (2 - \gamma)} = p - q$, so both terms are equal and $\delta' =
  \frac{1}{2 - \gamma} - p$. Substituting for the value of $\gamma$ gives
   \[ \delta' = \frac{1 - p - 2 q}{3} > 0 \]
  by the assumed inequality $p + 2 q < 1$.
\end{proof}

\section{Conditional lower bounds on 3SUM and 3XOR}

In this section we prove our main result, which we have reproduced below for the reader's convenience. 

\restateonlinetheorem*
\begin{proof}
    There is an efficient, $n^{1-\gamma}$-almost balanced, almost linear, and pairwise-independent hash family for $(\mathbb{Z}_{2 n^3}, +)$ for any $0 < \gamma < 1$ \cite{KPP16}. Applying \Cref{thm:clb-framework-online} gives the stated conditional lower from 3SUM. 

    There is an efficient,  $n^{1-\gamma}$-almost balanced, perfectly linear, and pairwise-independent hash family for $(\mathbb{Z}_{2}^{3\log n}, \oplus)$ for any $0 < \gamma < 1$ \cite{JV16}. Applying \Cref{thm:clb-framework-online} gives the stated conditional lower from 3XOR. 
\end{proof}

\bibliographystyle{alpha}
\bibliography{references}

\end{document}